\documentclass[aps,nopacs,11pt]{revtex4}
\usepackage{amsmath}
\usepackage{graphicx,,epsfig,amsmath,latexsym,amssymb,verbatim,color}
\usepackage{dsfont}
\usepackage{MnSymbol}

 \usepackage{hyperref}
\hypersetup{colorlinks=true,citecolor=blue,linkcolor=blue,filecolor=blue,urlcolor=blue,breaklinks=true}

\usepackage{theorem}\usepackage{url}
\newtheorem{definition}{Definition}
\newtheorem{proposition}[definition]{Proposition}
\newtheorem{lemma}[definition]{Lemma}

\newtheorem{theorem}[definition]{Theorem}

\def\squareforqed{\hbox{\rlap{$\sqcap$}$\sqcup$}}
\def\qed{\ifmmode\squareforqed\else{\unskip\nobreak\hfil
\penalty50\hskip1em\null\nobreak\hfil\squareforqed
\parfillskip=0pt\finalhyphendemerits=0\endgraf}\fi}
\def\endenv{\ifmmode\;\else{\unskip\nobreak\hfil
\penalty50\hskip1em\null\nobreak\hfil\;
\parfillskip=0pt\finalhyphendemerits=0\endgraf}\fi}
\newenvironment{proof}{\noindent \textbf{{Proof~} }}{\qed}

\mathchardef\ordinarycolon\mathcode`\:
\mathcode`\:=\string"8000
\def\vcentcolon{\mathrel{\mathop\ordinarycolon}}
\begingroup \catcode`\:=\active
  \lowercase{\endgroup
  \let :\vcentcolon
  }

\newcommand{\nc}{\newcommand}
\nc{\rnc}{\renewcommand}
\nc{\beg}{\begin{equation}}
\nc{\eeq}{{\end{equation}}}
\nc{\beqa}{\begin{eqnarray}}
\nc{\eeqa}{\end{eqnarray}}
\nc{\lbar}[1]{\overline{#1}}
\nc{\bra}[1]{\langle#1|}
\nc{\ket}[1]{|#1\rangle}
\nc{\ketbra}[2]{|#1\rangle\!\langle#2|}
\nc{\braket}[2]{\langle#1|#2\rangle}

\nc{\proj}[1]{| #1\rangle\!\langle #1 |}
\nc{\avg}[1]{\langle#1\rangle}
\nc{\Rank}{\operatorname{Rank}}
\nc{\smfrac}[2]{\mbox{$\frac{#1}{#2}$}}
\nc{\tr}{\operatorname{Tr}}
\nc{\ox}{\otimes}
\nc{\dg}{\dagger}
\nc{\dn}{\downarrow}
\nc{\cA}{{\cal A}}
\nc{\cB}{{\cal B}}
\nc{\cC}{{\cal C}}
\nc{\cD}{{\cal D}}
\nc{\cE}{{\cal E}}
\nc{\cF}{{\cal F}}
\nc{\cG}{{\cal G}}
\nc{\cH}{{\cal H}}
\nc{\cI}{{\cal I}}
\nc{\cJ}{{\cal J}}
\nc{\cK}{{\cal K}}
\nc{\cL}{{\cal L}}
\nc{\cM}{{\cal M}}
\nc{\cN}{{\cal N}}
\nc{\cO}{{\cal O}}
\nc{\cP}{{\cal P}}
\nc{\cQ}{{\cal Q}}
\nc{\cR}{{\cal R}}
\nc{\cS}{{\cal S}}
\nc{\cT}{{\cal T}}
\nc{\cX}{{\cal X}}
\nc{\cY}{{\cal Y}}
\nc{\cZ}{{\cal Z}}
\nc{\cW}{{\cal W}}
\nc{\csupp}{{\operatorname{csupp}}}
\nc{\qsupp}{{\operatorname{qsupp}}}
\nc{\var}{{\operatorname{var}}}
\nc{\rar}{\rightarrow}
\nc{\lrar}{\longrightarrow}
\nc{\polylog}{{\operatorname{polylog}}}
\nc{\1}{{\mathds{1}}}
\nc{\wt}{{\operatorname{wt}}}
\nc{\av}[1]{{\left\langle {#1} \right\rangle}}
\nc{\supp}{{\operatorname{supp}}}

\def\a{\alpha}

\def\G{\Gamma}

\def\S{\Sigma}
\def\U{\Upsilon}

\nc{\RR}{{{\mathbb R}}}
\nc{\CC}{{{\mathbb C}}}
\nc{\FF}{{{\mathbb F}}}
\nc{\NN}{{{\mathbb N}}}
\nc{\ZZ}{{{\mathbb Z}}}
\nc{\PP}{{{\mathbb P}}}
\nc{\QQ}{{{\mathbb Q}}}
\nc{\UU}{{{\mathbb U}}}
\nc{\EE}{{{\mathbb E}}}
\nc{\id}{{\operatorname{id}}}
\newcommand{\cvartheta}{{\widetilde\vartheta}}
\nc{\CHSH}{{\operatorname{CHSH}}}

\nc{\be}{\begin{equation}}
\nc{\ee}{{\end{equation}}}
\nc{\bea}{\begin{eqnarray}}
\nc{\eea}{\end{eqnarray}}
\nc{\<}{\langle}
\rnc{\>}{\rangle}
\nc{\Hom}[2]{\mbox{Hom}(\CC^{#1},\CC^{#2})}
\nc{\rU}{\mbox{U}}

\nc{\ob}[1]{#1}
\nc{\SEP}{{\text{SEP}}}
\nc{\NS}{{\text{NS}}}
\nc{\LOCC}{{\text{LOCC}}}
\nc{\PPT}{{\text{PPT}}}
\nc{\EXT}{{\text{EXT}}}
\nc{\Sym}{{\operatorname{Sym}}}
\nc{\ERLO}{{E_{\text{r,LO}}}}
\nc{\ERLOCC}{{E_{\text{r,LOCC}}}}
\nc{\ERPPT}{{E_{\text{r,PPT}}}}
\nc{\ERLOCCinfty}{{E^{\infty}_{\text{r,LOCC}}}}
\nc{\Aram}{{\operatorname{\sf A}}}

\begin{document}

\title{Separation between quantum Lov\'asz number and entanglement-assisted zero-error classical capacity}

\author{Xin Wang$^{1}$}
\email{xin.wang-8@student.uts.edu.au}
\author{Runyao Duan$^{1,2}$}
\email{runyao.duan@uts.edu.au}

\affiliation{$^1$Centre for Quantum Software and Information, School of Software, Faculty of Engineering and Information Technology, University of Technology Sydney, NSW 2007, Australia}
\affiliation{$^2$UTS-AMSS Joint Research Laboratory for Quantum Computation and Quantum Information Processing, Academy of Mathematics and Systems Science, Chinese Academy of Sciences, Beijing 100190, China}

\thanks{A preliminary version of this paper was presented as a contributed talk at the 16th Asian Quantum Information Science Conference (AQIS'16).}

\begin{abstract}
Quantum Lov\'asz number is a quantum generalization of the Lov\'asz number in graph theory. It is the best known efficiently computable upper bound of the entanglement-assisted zero-error classical capacity of a quantum channel. However, it remains an intriguing open problem whether quantum entanglement can always enhance the zero-error capacity to achieve the quantum Lov\'asz number. In this paper, by constructing a particular class of qutrit-to-qutrit channels, we show that there exists a strict gap between the entanglement-assisted zero-error capacity and the quantum Lov\'asz number. 
Interestingly, for this class of quantum channels, the quantum generalization of fractional packing number is strictly larger than the zero-error capacity assisted with feedback or no-signalling correlations, which differs from the case of classical channels.
\end{abstract}
\maketitle

\section{Introduction}
A fundamental problem of information theory is to determine the capability of a communication channel to deliver messages from the sender to the receiver. While the conventional information theory focuses on sending messages with asymptotically vanishing errors~\cite{Shannon1948}, Shannon also investigated this problem in the zero-error setting and described the zero-error capacity of a channel as the maximum rate at which it can be used to transmit information with a zero probability of error~\cite{Shannon1956}. Zero-error information theory~\cite{Shannon1956,Korner1998} concerns the asymptotic combinatorial problems, most of which are difficult and unsolved. 

Recently the zero-error information theory has been studied in the quantum setting and many new phenomena were observed.  One remarkable result is the super-activation of the zero-error classical/quantum capacities of quantum channels~\cite{Duan2008a, Duan2009, Cubitt2011a, Cubitt2012, Shirokov2015a}. Another important result is that, for some classical channels, quantum entanglement can be used to improve the zero-error capacity~\cite{Cubitt2010, Leung2012}, while there is no such advantage for the normal capacity~\cite{Shannon1948}.  Furthermore, there are more kinds of capacities when considering auxiliary resources,  such as the shared entanglement~\cite{Duan2008a, Cubitt2010,  Leung2012, Duan2013, Briet2015, Piovesan2015,Stahlke2014}, the no-signalling correlations~\cite{Matthews2012,Cubitt2011,Duan2016,Leung2015c,Xie2017a,Wang2016g,Wang2017d,Wang2017a},  and the feedback assistance~\cite{Shannon1956,Duan2015}. All of these capacities are only partially understood, and the zero-error information theory of quantum channels seems more complex than that of classical channels.

To study the zero-error communication via quantum channels, the so-called ``non-commutative graph theory''  was introduced in~\cite{Duan2013}. The non-commutative graph (an object based on an operator system) associated with a quantum channel fully captures the zero-error communication properties of this channel~\cite{Duan2013}, thus playing a similar role to confusability graph in the classical case. It is well-known that the zero-error capacity is extremely difficult to compute for both classical and quantum channels~\cite{Beigi2007}. Nevertheless, the zero-error capacity of a classical channel is upper bounded by the Lov\'asz number~\cite{Lovasz1979} while the zero-error capacity of a quantum channel is upper bounded by the quantum Lov\'asz number~\cite{Duan2013}. Furthermore, the entanglement-assisted zero-error capacity of a classical channel is also upper-bounded by the Lov\'asz  number~\cite{Beigi2010,Duan2013}, and this result can be generalized to quantum channels by using the quantum Lov\'asz  number~\cite{Duan2013}. 

A more general problem is the simulation of a channel, which concerns how to use a channel $\cN$ from Alice ($A$) to Bob ($B$) to simulate another channel $\cM$ also from $A$ to $B$~\cite{Kretschmann2004}. Shannon's noisy channel coding theorem determines the capability of any noisy channel $\cN$ to simulate a noiseless channel~\cite{Shannon1948} and the reverse Shannon theorem was proved in~\cite{Bennett2002}.  The quantum reverse Shannon theorem was proved recently~\cite{Bennett2014, Berta2011a}, which states that any quantum channel can be simulated by an amount of classical communication equal to its entanglement-assisted capacity assisted with free entanglement. In the zero-error setting, there is a kind of reversibility between the zero-error capacity and  simulation cost in the presence of no-signalling correlations~\cite{Cubitt2011}. More recently, the no-signalling-assisted (NS-assisted) zero-error simulation cost of a quantum channel was introduced in~\cite{Duan2016}.

An intriguing open problem in zero-error information theory is whether the entanglement-assisted zero-error capacity always coincides with the quantum  Lov\'asz number for a classical or quantum channel, which is frequently mentioned in~\cite{Leung2012,Duan2013,Beigi2010,Cubitt2011,Cubitt2014a,Mancinska2013}. If they are equal, it will imply that the entanglement-assisted zero-error capacity is additive, while the unassisted case is not~\cite{Alon1998}.

In this paper, we resolve the above open problem for quantum channels. To be specific, we construct a class of qutrit-to-qutrit channels for which the quantum  Lov\'asz number is strictly larger than the entanglement-assisted zero-error capacity. We utilize the one-shot NS-assisted zero-error capacity and simulation cost to determine the asymptotic NS-assisted zero-error capacity in this case, which is potentially larger than the entanglement-assisted zero-error capacity. An interesting fact is that this class of channels are reversible in a strong sense. To be specific, for this class of channels, the one-shot NS-assisted zero-error capacity and simulation cost are identical. We then give a closed formula for the quantum Lov\'asz number for this class of channels, and use it to conclude that there is a strict gap between the quantum Lov\'asz  number and the entanglement-assisted zero-error capacity.  For this class of channels, we also find that the quantum fractional packing number is strictly larger than the feedback-assisted or NS-assisted zero-error capacity, while these three quantities are equal to each other for any classical channel~\cite{Cubitt2011}.

\section{Preliminaries}
In the following, we will frequently use symbols such as $A$ (or $A'$) and $B$ (or $B'$) to denote the (finite-dimensional) Hilbert spaces associated with Alice and Bob, respectively. The set of linear operators over $A$ is denoted by $\cL(A)$.  
A quantum channel $\cN$ from $A$ to $B$ is simply a completely positive and trace-preserving (CPTP) linear map from $\cL(A)$ to $\cL(B)$,
with a Choi-Kraus operator sum representation
$\cN(\rho)=\sum_k E_k\rho E_k^\dag$. where $\sum_k E_k^\dag E_k=\1_{A}$. 
The Choi-Kraus operator space of $\cN$ is denoted by $$K=K(\cN):=\operatorname{span}\{E_k\}.$$ 
Such space is alternatively called ``non-commutative bipartite graph'' since it determines the zero-error capacity of a quantum channel in the presence of noiseless feedback~\cite{Duan2015}, which plays a similar role to the bipartite graph of a classical channel. 
The Choi-Jamio\l{}kowski matrix of $\cN: \cL(A')\to \cL(B)$ is  $J_{AB}=\sum_{ij} \ketbra{i}{j}_A \ox \cN(\ketbra{i}{j}_{A'})=(\text{id}_{A}\ox\cN)\proj{\Phi}$, where $A$ and $A'$ are isomorphic Hilbert spaces with respective orthonormal basis $\{\ket i_A\}$ and $\{\ket j_{A'}\}$, $\ket \Phi =\sum_i \ket{i}_A\ket{i}_{A'}$ and $\text{id}_{A}$ is the identity map. We denote $P_{AB}$  as the projection onto the support of $J_{AB}$, which is the subspace $(\1\ox K)\ket{\Phi}$.
The non-commutative graph~\cite{Duan2013} of $\cN$ is defined by the operator subspace 
$$S:=K^\dagger K=\text{span}\{E_j^\dagger E_k: j, k\}<\cL(A'),$$ 
where $S<\cL(A')$ means that $S$ is a subspace of $\cL(A')$. 

The one-shot zero-error capacity of a quantum channel $\cN$ is the maximum number of inputs such that the receiver can perfectly distinguish the corresponding output states. The output states can be perfectly distinguished if and only if they are orthogonal.
This one-shot zero-error capacity can be equivalently defined as  the  independence number $\a(S)$ of the non-commutative graph~\cite{Duan2013} of $\cN$, i.e.,  the maximum size of a set of orthogonal unit vectors $\{\ket{\phi_m}: m=1, ... , M \}$ such that $$\forall m\neq m',  \ketbra{\phi_m}{\phi_m'}\in S^{\perp}.$$ 
The zero-error capacity is given by regularization of $\a(S)$, i.e.,
\begin{equation}\label{C0}
C_0(\cN)=C_0(S)=\sup_{n \to \infty}\frac{1}{n}\log\a(S^{\ox n}).
\end{equation}
Throughout this paper, $\log$ denotes the binary logarithm $\log_2$. The $\sup$ in Eq. (\ref{C0}) can be replaced by $\lim$ based on  the lemma about existence of limits in~\cite{Barnum1998}.

The entanglement-assisted independence number $\widetilde \a (S)$~\cite{Duan2013} is motivated by the scenario where sender and receiver share entangled state beforehand and it quantifies the maximum number of distinguishable messages that can be sent via the channel $\cN$ with graph $S$ when shared entanglement is free. To be specific, $\widetilde \a (S)$ is the maximum integer $M$ such that there exist Hilbert spaces $A_0, B_0$ and a state $\sigma\in \cL(A_0\ox B_0)$, and CPTP maps $\cE_m: \cL(A_0) \to \cL(A)  (m=1, ... , N)$ such that  the $N$ output states   $\rho_m=(\cN \circ \cE_m \ox \text{id}_{B_0})\sigma$  are orthogonal.
The entanglement-assisted  zero-error capacity of $S$ is given by regularization of $\widetilde\a(S)$, i.e.,
\begin{equation}
C_{\rm{0E}}(\cN)=C_{\rm{0E}}(S)=\sup_{n \to \infty}\frac{1}{n}\log\widetilde\a(S^{\ox n}).
\end{equation}

For any non-commutative graph $S < \mathcal{L}(A)$, the quantum Lov\'asz number $\cvartheta(S)$ was introduced as
 a quantum analog of the Lov\'asz  number in~\cite{Duan2013}.  It can be formalized by semidefinite programming (SDP)~\cite{Duan2013} as follows: 
\begin{align}
  \label{eq:c-theta-SDP}
  \cvartheta(S) &= \max\ \bra{\Phi} (\1\ox\rho + T) \ket{\Phi} \\
          & \phantom{==}\text{s.t. }\ T \in S^\perp\ox\mathcal{L}(A'), \quad \tr\rho = 1, \nonumber \\
          & \phantom{==\text{s.t. }}\ \1\ox\rho+T \geq 0, \quad \rho \geq 0, \nonumber
\end{align}
where $\ket \Phi=\sum_{i}\ket {i}_{A}\ket{i}_{A'}$. Note that
SDP can be solved by polynomial-time algorithms~\cite{Vandenberghe1996, Khachiyan1980} in usual and it has many other applications in quantum information theory (e.g.,~\cite{Rains2001,Wang2016a,Barnum2003,Wang2016d,Eldar2003,Li2017,Watrous2012,Regev2013}).
More details about SDP can be found in~\cite{Watrous2011b}.
The dual  SDP of  $\cvartheta(S)$ is given by
\begin{equation}\begin{split}
  \label{eq:c-theta-SDP-dual}
  \cvartheta(S) &= \min\ \| \tr_A Y \|_\infty \\
                & \phantom{==}\text{s.t. }\ Y \in S \ox \mathcal{L}(A'), \quad
                                            Y \geq \proj\Phi.
\end{split}\end{equation}
The operator norm $\|R\|_\infty$ is defined as the maximum eigenvalue of $\sqrt{R^\dagger R}$.
By strong duality, the optimal values of the primal and dual SDPs of  $\cvartheta(S)$ coincide.  
Furthermore,   $\cvartheta(S)$ was proved to be an upper bound of $C_{\rm{0E}}(S)$~\cite{Duan2013},
\begin{equation}
C_0(S)\le C_{\rm{0E}}(S) \le \log\cvartheta(S).
\end{equation}
Moreover, for a quantum channel $\cN$ with non-commutative graph $S$, the quantum Lov\'asz number of $\mathcal{N}$ is naturally given by the quantum Lov\'asz number of $S$,
$$\cvartheta(\cN)=\cvartheta(S).$$

The no-signalling correlations arises in the research of the relativistic causality of quantum operations~\cite{Beckman2001, Eggeling2002a, Piani2006, Oreshkov2012} and Cubitt et al.~\cite{Cubitt2011} first introduced classical no-signalling correlations into the zero-error communication via classical channels and proved that the fractional packing number of the bipartite graph induced by the channel equals to the zero-error capacity of the channel.  Recently,  quantum no-signalling correlations were introduced  into the zero-error communication via quantum channels in~\cite{Duan2016} and the one-shot NS-assisted zero-error classical capability (quantified as the number of messages) was formulated as the following SDP:
\begin{equation}\begin{split}\label{eq:Upsilon}
\U(\cN)=\U(K) = \max &\tr R_A \\
\text{ s.t. }&  0 \leq U_{AB} \leq R_A \ox \1_B, \\
        & \tr_A U_{AB} = \1_B, \\
        & \tr P_{AB}(R_A\ox\1_B-U_{AB}) = 0,
\end{split}\end{equation}
where $P_{AB}$ denotes the projection onto  $(\1\ox K)\ket{\Phi}$. The asymptotic NS-assisted zero-error capacity  is given by the regularization:
\begin{equation}
C_{0,\rm{NS}}(\cN)=C_{0,\rm{NS}}(K)=\sup_{n \to \infty}\frac{1}{n}\log\U(K^{\ox n}).
\end{equation}
A remarkable feature of NS-assisted zero-error capacity is that one bit noiseless communication can fully activate any classical-quantum channel to achieve its asymptotic capacity~\cite{Duan2015a}.

The zero-error simulation cost of a quantum channel in the presence of quantum no-signalling correlations was introduced  in~\cite{Duan2016} and  formalized as SDPs. To be specific, for the quantum channel $\cN$ with Choi-Jamio\l{}kowski matrix $J_{AB}$, the NS-assisted zero-error simulation cost of $\cN$ is given by
\begin{equation}
S_{\rm{0, NS}}(\cN)=-H_{\min}(A|B)_{J_{AB}}:=\log \S(\cN),
\end{equation}
where
\begin{equation}\begin{split}\label{SC SDP}
 \S(\cN)= \min& \tr T_B,\\
  {\rm s.t.} &\ J_{AB}\leq \1_A\ox T_B,
\end{split}\end{equation}
 and $H_{\min}(A|B)_{J_{AB}}$ is the so-called conditional  min-entropy~\cite{Konig2009,Tomamichel2012}.
By the fact that the conditional  min-entropy is additive~\cite{Konig2009}, the asymptotic NS-assisted zero-error simulation cost  is given by 
\begin{equation}
S_{\rm{0,NS}}(\cN)=\log \S(\cN).
\end{equation}

Furthermore, noting that the NS assistance is stronger than the entanglement assistance, the capacities and simulation cost of a quantum channel introduced above obey the following inequality:
\begin{equation}\label{simulation}
C_{0}\le C_{\rm{0E}}\le C_{\rm{0, NS}}\le C_{\rm E}\le S_{\rm{0, NS}},
\end{equation}
where $C_{\rm E}$ is the entanglement-assisted classical capacity~\cite{Bennett2002}.

\section{Gap between quantum Lov\'asz number and entanglement-assisted zero-error capacity}
In this section, we are going to show the gap between the quantum Lov\'asz number and the entanglement-assisted zero-error capacity.
The difficulty in comparing $C_{\rm{0E}}$ and the quantum Lov\'asz number is that there are few channels whose entanglement-assisted zero-error capacity is known. In fact,  $C_{\rm{0E}}$ is even not known to be computable.  The problem whether there exists a gap between them was a prominent open problem in the area of zero-error quantum information theory.

Our approach to the above problem is to construct a particular class of channels and considering the NS-assisted zero-error capacity, which is potentially larger than the entanglement-assisted case. 
To be specific, the  class of channels we use is $\cN_{\alpha}(\rho)=E_\alpha\rho E_\alpha^{\dagger}+D_\alpha\rho D_\alpha^{\dagger}$ $(0<\alpha\le \pi/4)$ with 
\begin{align*}
E_\alpha = \sin \alpha\ketbra{0}{1}+\ketbra{1}{2},\\
D_\alpha=\cos\alpha\ketbra{2}{1}+\ketbra{1}{0}.
\end{align*}
This qutrit-qutrit channel $\cN_\a$ is motivated in the similar sipirt of the amplitutde damping channel, which exhibits a significant differnece from the classical-quantum channels.
 
The first Choi-Kraus operator $E_\a$ annihilates the ground state $\proj 0$:
$$E_\a \proj 0 E_\a^\dagger=0,$$ and it decays the state $\proj 1$ to the ground state $\proj 0$: 
$$E_\a \proj 1 E_\a^\dagger=\sin^2\a\proj 0.$$ Meanwhile, $E_\a$ also transfer the state $\proj 2$ to $\proj 1$, i.e.,  $E_\a \proj 2 E_\a^\dagger=\proj 1$. 
On the other hand, the choice of $D_\a$ above ensures that 
$$E_\a^\dagger E_\a+ D_\a^\dagger D_\a=\1,$$
which means that the operators $E_\a$ and $D_\a$ are valid Kraus operators for a quantum channel.

The Choi-Jamio\l{}kowski matrix of $\cN_{\alpha}$ is given by
\begin{align*}
J_\alpha=(1+\sin^2\alpha)\proj{u_\alpha}+(1+\cos^2\alpha)\proj{v_\alpha},
\end{align*}
where 
\begin{align}
\ket {u_\alpha}=\frac{\sin\alpha}{\sqrt{1+\sin^2\alpha}}\ket {10}+\frac{1}{\sqrt{1+\sin^2\alpha}}\ket {21},\\
 \ket{v_\alpha}=\frac{\cos\alpha}{\sqrt{1+\cos^2\alpha}}\ket {12}+\frac{1}{\sqrt{1+\cos^2\alpha}}\ket {01}.
\end{align}
Then, the projection onto  the support of $J_\alpha$ is
\begin{equation}\begin{split}
P_\alpha=\proj{u_\alpha}+\proj{v_\alpha}.
\end{split}\end{equation}

We first prove that both NS-assisted zero-error capacity and simulation cost of $\cN_\a$ are exactly two bits.
\begin{proposition}\label{rev N}
For the channel $\cN_{\alpha}$ $(0< \alpha \le \pi/4)$, 
\begin{equation}
C_{0,\rm{NS}}(\cN_{\alpha})=C_{\rm E}(\cN_\a)= S_{\rm{0,NS}}(\cN_{\alpha})=2.
\end{equation}	
\end{proposition}
\begin{proof}
First, we show that Alice can trasmit at least 2 bits prefectly to Bob with a single use of $\cN_\a$ and the NS-assistance. The approach is to construct a feasible solution of the SDP~\eqref{eq:Upsilon} of the one-shot NS-assisted zero-error capacity. To be specific, suppose that $R_A=2(\cos^2\alpha\proj{0}+\proj{1}+\sin^2\alpha\proj {2})$ and
\begin{align*}
U_{AB}=&\cos^2\alpha\proj{01}+\sin^2\a\proj{21}+\proj{10}+\proj{12}\\
&+\sin\alpha(\ketbra{10}{21}+\ketbra{21}{10})+\cos\alpha(\ketbra{01}{12}+\ketbra{12}{01}).
\end{align*}
One can simply check  that $R_{A}\ox \1_B- U_{AB}\ge 0$, $\tr_A U_{AB}=\1_B$ and $P_\a (R_{A}\ox \1_B- U_{AB})=0$.
Therefore, $\{R_A, U_{AB}\}$ is a feasible solution to SDP (\ref{eq:Upsilon}) of $\U(\cN_\alpha)$, which means that 
 \begin{equation}\label{C N alpha}
 C_{0,\rm{NS}}(\cN_\alpha)\ge\log\U(\cN_\alpha)\ge \log \tr R_A=2.
\end{equation}

Second, we prove that the one-shot NS-assisted simulation cost of $\cN_\a$ is at amost 2 bits. We utilize the SDP \eqref{SC SDP} of  one-shot NS-assisted simulation cost  and choose
 \begin{equation}
 T_B=2(\sin^2\alpha\proj{0}+\proj{1}+\cos^2\alpha\proj {2}).
 \end{equation}
It can be checked that
$\1\ox T_B - J_\alpha \ge 0$. Thus, $T_B$ is a feasible solution to SDP (\ref{SC SDP}) of $\S(\cN_\alpha)$, which means that 
\begin{equation}\label{S N alpha}
S_{\rm{0,NS}}(\cN_\alpha)\le\log\S(\cN_\alpha)\le \log\tr T_B=2.
\end{equation}

Finally, combining Eq. (\ref{C N alpha}), Eq. (\ref{S N alpha}) and Eq. (\ref{simulation}), it is clear that 
\begin{equation}
C_{0,\rm{NS}}(\cN_{\alpha})=C_{\rm E}(\cN_\a)=S_{\rm{0,NS}}(\cN_{\alpha})=2.
\end{equation}
\end{proof}

We then solve the exact value of the quantum Lov\'asz number  of $\cN_\a$.
\begin{proposition}\label{theta N}
 For the channel $\cN_{\alpha}$ $(0< \alpha \le \pi/4)$,
\begin{equation}
\cvartheta(\cN_{\alpha})=2+\cos^2\alpha+{\cos^{-2}\alpha}>4.
\end{equation}
\end{proposition}
\begin{proof}
We first construct a quantum state $\rho$ and an operator $T\in S^\perp\ox\mathcal{L}(A')$ such that $\1\ox\rho+T$ is positive semidefinite. Then, we use the  primal SDP~\eqref{eq:c-theta-SDP} of $\cvartheta(\cN_\alpha)$ to obtain the lower bound of  $\cvartheta(\cN_\alpha)$. 

To be specific, the non-commutative graph of $\cN_\alpha$ is $S=\text{span}\{F_1, F_2, F_3, F_4\}$ with
\begin{align}
F_1&=\proj{0}+\cos^2\a\proj{1},\\
F_2&=\sin^2\a\proj{1}+\proj{2},\\
F_3&=\ketbra{0}{2} \text{ and } F_4=\ketbra{2}{0}.\\
\end{align}
Let us choose
\begin{equation}
\rho=\frac{\cos^2\a}{1+\cos^2\a}\proj{0}+\frac{1}{1+\cos^2\a}\proj{1}
\end{equation}
and
$T=T_1\ox T_2+R$,
where 
\begin{align}
T_1&=\frac{1}{1+\cos^2\a}(\proj{0}-\frac{1}{\cos^2\a}\proj 1+\frac{\sin^2\a}{\cos^2\a}\proj{2}),  \\
T_2&=\cos^4\a\proj 0-\proj{1}, \\
 R&=\ketbra{00}{11}+\ketbra{11}{00}.
\end{align}
It is clear that $\rho\ge 0$ and $\tr \rho =1$.  Also, it is easy to see that for any matrix $M\in\mathcal{L}(A')$ and $j=1,2,3,4$, 
\begin{equation}
\tr R(F_j\ox M)=0.
\end{equation}
Meanwhile, noticing that $\tr(T_1F_j)=0$ for $j=1,2,3,4$,
we have
\begin{equation}
T=T_1\ox T_2+R \in S^\perp\ox\mathcal{L}(A').
\end{equation}
Moreover, it is easy to see that 
\begin{equation}\begin{split}
\1\ox\rho+T=&\cos^2\a\proj{00}+\frac{1}{\cos^2\a}\proj{11}+\ketbra{00}{11}\\
&+\ketbra{11}{00}+\frac{\cos^2\a-\cos^4\a}{1+\cos^2\a}\proj{20}\\
&+\frac{2\cos^2-1}{(1+\cos^2\a)\cos^2\a}\proj{21}\ge 0.
\end{split}\end{equation}
Then, $\{\rho, T\}$ is a feasible solution to primal SDP (\ref{eq:c-theta-SDP}) of $\cvartheta(\cN_\alpha)$. Hence, we have that 
\begin{equation}\label{theta lower bound}
\begin{split}
\cvartheta(\cN_\a)&\ge \tr[\proj{\Phi} (\1\ox\rho + T)]\\
&= \tr[\proj{\Phi} (\1\ox\rho + T_1\ox T_2+R)]\\
&=2+\cos^2\alpha+{\cos^{-2}\alpha}.
\end{split}\end{equation}

On the other hand, we find a feasible solution to the dual SDP (\ref{eq:c-theta-SDP-dual}) of $\cvartheta(\cN_\alpha)$. It is easy to see that 
\begin{equation}
S^\perp =\rm span\{M_1, M_2, M_3, M_4, M_5 \},
\end{equation}
 where
$M_1=\ketbra{0}{1}$, $M_2=\ketbra{1}{0}$, $M_3=\ketbra{1}{2}$, $M_4=\ketbra{2}{1}$ and $M_5=\ketbra{0}{0}-\cos^{-2}\a\proj{1}+\tan^2\a\proj 2$.
Let us choose
\begin{equation}
Y=Y_1 \ox(\proj{0}+\proj{1})+Y_2\ox \proj{2}+\frac{1+\cos^2\alpha}{\cos^2\alpha}Y_3
\end{equation}
with
\begin{align} 
Y_1=&(1+\cos^2\alpha)\cos^{-2}\alpha\proj{0}+(1+\cos^2\a)\proj 1,\\
Y_2=&(2-{\cos^{-2}\a})\proj{0}+({\cos^{-2}\a}-\sin^2\a)\proj{1}\\
&+(1+\cos^2\a){\cos^{-2}\a}\proj2,\\
Y_3=&\ketbra{00}{22}+\ketbra{22}{00}.
\end{align}
It is easy to see that for any matrix $V\in\mathcal{L}(A')$ and $j=1,2,3,4,5$, we have that
\begin{equation}
\tr Y_3(M_j\ox V)=0.
\end{equation}
Meanwhile, since $\tr(Y_kM_j)=0$ for $k=1,2$ and $j=1,2,3,4,5$, 
we have that 
\begin{align*}
Y&=Y_1 \ox (\proj{0}+\proj{1})+Y_2\ox \proj{2}+\frac{1+\cos^2\alpha}{\cos^2\alpha}Y_3 \\
&\in S\ox\mathcal{L}(A').
\end{align*}
It is also easy to check that $Y-\proj{\Phi}\ge 0$. Thus, $Y$ is a feasible solution to SDP (\ref{eq:c-theta-SDP-dual}) of  $\cvartheta(\cN_\alpha)$. 
Furthermore, one can simply  calculate that
\begin{equation}
\tr_A Y=(2+\cos^2\a+\cos^{-2}\a)\1_B,
\end{equation}
Therefore, 
 \begin{equation}\label{theta upper bound}
 \cvartheta(\cN_\alpha)\le \| \tr_A Y \|_\infty= 2+\cos^2\a+\cos^{-2}\a.
 \end{equation}

 Finally, combining Eq. (\ref{theta lower bound}) and Eq. (\ref{theta upper bound}), we can conclude that
 $$\cvartheta(\cN_\alpha)=2+\cos^2\a+\cos^{-2}\a.$$
\end{proof}

Now we are able to show a separation between $\log\cvartheta(\cN_\a)$ and  $C_{\rm{0E}}(\cN_\a)$.
\begin{theorem}
For the channel $\cN_\a$  $(0< \alpha \le \pi/4)$, the quantum Lov\'asz number is strictly larger than the
 entanglement-assisted zero-error capacity (or even with no-signalling assistance), i.e.,
\begin{equation}
\log\cvartheta(\cN_\a) >C_{0,\rm{NS}}(\cN_\a)\ge C_{\rm{0E}}(\cN_\a).
\end{equation}

\end{theorem}
\begin{proof}
It is easy to see this result from Proposition \ref{rev N} and Proposition \ref{theta N}. To be specific, we have
\begin{align}
\log\cvartheta(\cN_\a)&=\log(2+\cos^2\a+\cos^{-2}\a)\\
&>2
=C_{0,\rm{NS}}(\cN_\a)\\
&\ge  C_{\rm{0E}}(\cN_\a).
\end{align}
\end{proof}

\section{Gap between quantum fractional packing number and feedback-assisted or NS-assisted zero-error capacity}
 A classical channel $\cN=(X,p(y|x),Y)$ naturally induces a bipartite graph $\Gamma(\cN)=(X,E,Y)$, where $X$ and $Y$ are the input and output alphabets, respectively.  And $E\subset X\times Y$ is the set of edges such that $(x,y)\in E$ if and only if the probability $p(y|x)$ is positive.
The non-commutative bipartite graph in this case is given by
 $$K = \text{span}\{\ketbra{y}{x} : (x, y) \text{ is an edge in } \Gamma \}.$$
 
Shannon first introduced the feedback-assisted zero-error capacity~\cite{Shannon1956}.
To be precise, his model has noiseless instantaneous feedback of the channel output back to the sender, and it requires some arbitrarily small rate of forward noiseless communication. For any classical channel with a positive zero-error capacity, he showed that the feedback-assisted zero-error capacity $C_{\rm{0F}}$ of a classical channel $\cN$ is given by the fractional packing number of its bipartite graph~\cite{Shannon1956}:
 $${\alpha ^*}(\G) = \max \sum\limits_x {v_x} \ \text{  s.t. }\sum\limits_x {v_x {p(y|x)}  }  \le 1 \forall y, 0\le v_x \le 1 \forall x.$$
For any classical bipartite graph, the fractional packing number also gives the NS-assisted zero-error classical capacity and simulation cost~\cite{Cubitt2011}, i.e.,
$$C_{0,\rm{NS}}(K)=S_{\rm{0,NS}}(K)=\log\a^*(\G).$$

The quantum generalization of  fractional packing number in \cite{Duan2016} was suggested by Harrow as \begin{equation}\begin{split}
  \label{eq:Aram}
  \Aram(K) &= \max\tr R_A  
\text{ s.t. }  0 \leq R_A,
  \tr_A P_{AB}(R_A\ox\1_B) \leq \1_B, \\
&= \min\tr T_B  \text{ s.t. } 0 \leq T_B,
  \tr_B P_{AB}(\1_A\ox T_B) \geq \1_A. \\
\end{split}\end{equation}
This quantum fractional packing number $\Aram(K)$ has nice mathematical properties such as additivity under tensor product~\cite{Duan2016}. 

For any bipartite graph $\G$, quantum fractional packing number also reduces to the fractional packing number, i.e.,
\begin{equation}
\Aram(K)=\a^*(\G).
\end{equation}
Furthermore, for a classical-quantum channel with non-commutative bipartite graph $K$, it also holds that~\cite{Duan2016}
\begin{equation}
C_{0,\rm{NS}}(K)=\log \Aram(K).
\end{equation}

However, if we consider general quantum channels, this quantum fractional packing number will exceed the NS-assisted zero-error capacity as well as the feedback-assisted zero-error capacity. An example is the class of channels $\cN_\a$ and the proof is in the following Proposition \ref{aram number gap}. For $\cN_\a$, it is easy to see that
the set of linear operators $\{E_i^\dagger E_j\}$ is linearly independent, which means that
$\cN_\a$ is an extremal channel~\cite{Choi1975}. Thus, its non-commutative bipartite graph $K_\a$ is an extremal graph~\cite{Duan2016}, which means that there can
only be a unique channel $\cN$ such that $K(\cN)=K_\a$.

For a general quantum channel, its feedback-assisted zero-error capacity  depends only on its non-commutative bipartite graph. And the feedback-assisted zero-error capacity is always smaller than or equal to the entanglement-assisted classical capacity~\cite{Duan2015}, i.e., 
\begin{equation}
C_{\rm{0F}}(K) \le C_{\text{minE}}(K),
\end{equation}
where $C_{\text{minE}}(K)$  is defined by
\begin{equation}
C_{\text{minE}}(K):=\min\{C_{\rm E}(\cN): K(\cN)<K\}.
\end{equation}
Considering the fact that $C_{0,\rm{NS}}(K)\le  C_{\text{minE}}(K)\le S_{\rm{0,NS}}(K)$~\cite{Duan2015}, it is easy to see that  $C_{\text{minE}}(K_\a)$ is exactly two bits from Proposition \ref{rev N}.

\begin{lemma}\label{N aram}
 For  non-commutative bipartite graph $K_\a$  $(0<\alpha \le \pi/4)$, the  quantum fractional  packing number is given by
\begin{equation}
\Aram(K_{\alpha})=2+\cos^2\alpha+{\cos^{-2}\alpha}.
\end{equation}
\end{lemma}
\begin{proof}
Let us choose $R_A=(2-\sin^2\alpha)\proj{0}+x\proj{1}$, then
\begin{align*}
\tr_A P_\alpha(R_A\ox\1_B)=\frac{x\sin^2\alpha}{1+\sin^2\alpha}\proj{0}+\proj{1}+\frac{x\cos^2\alpha}{1+\cos^2\alpha}\proj{2}.
\end{align*}
When $x=1+\cos^{-2}\alpha$, it is clear that $\tr_A P_\alpha(R_A\ox\1_B)\le \1_B$. Therefore,
$R_A$ is a feasible solution to the primal SDP of $\Aram(\cN_{\alpha})$, which means that 
\begin{equation}
\Aram(\cN_{\alpha})\ge \tr R_A= 2+\cos^2\alpha+{\cos^{-2}\alpha}.
\end{equation}

Similarly, it is easy to check that $T_B=(2-\sin^2\alpha)\proj{1}+({1+\cos^{-2}\alpha})\proj{2}$ is a feasible solution to the dual SDP of $\Aram(\cN_{\alpha})$. Therefore, 
\begin{equation}
\Aram(\cN_{\alpha})\le \tr T_B= 2+\cos^2\alpha+{\cos^{-2}\alpha}.
\end{equation}

Hence, we have that $\Aram(\cN_{\alpha})=2+\cos^2\alpha+{\cos^{-2}\alpha}$.
\end{proof}

Now, we are able to show the separation.
\begin{proposition}\label{aram number gap}
For  non-commutative bipartite graph $K_\a$  $(0<\alpha \le \pi/4)$, we have that
\begin{align} 
{C_{\rm{0F}}(K_\a)} < \log\Aram(K_\a),\\
{C_{0,\rm{NS}}(K_\a)} < \log\Aram(K_\a).
\end{align}
\end{proposition}
\begin{proof}
For general non-commutative bipartite graph $K$, it holds that $C_{\rm{0F}}(K)\le C_{\min \rm{E}}(K)$~\cite{Duan2015}. Then, by Proposition  \ref{rev N} and Lemma \ref{N aram},
we have
\begin{equation}
{C_{\rm{0F}}(K_\a)} \le {C_{\min\rm{E}}(K_\a)}=2 <\log \Aram(K_\a).
\end{equation}

From  Proposition \ref{rev N}  and Lemma \ref{N aram}, it is also clear that ${C_{0,\rm{NS}}(K_\a)} < \log\Aram(K_\a)$.
\end{proof}


\section{Discussions}
Interestingly, for the channel $\cN_\a$, its quantum fractional packing number is equal to its quantum Lov\'asz number.
Let us recall  that the Lov\'asz number of a classical graph $G$ has an  operational interpretation~\cite{Duan2016} as 
\begin{equation}
\vartheta(G)=\min\{\Aram(K): K^\dagger K<S_G \},
\end{equation}
where the minimization is over classical-quantum graphs $K$ and $S_G$ is non-commutative graph associated with $G$.
A natural and interesting question is that for the non-commutative graph $S$, do we have
\begin{equation}
\cvartheta(S)=\min\{\Aram(K): K^\dagger K<S\}?
\end{equation}

The non-commutative bipartite graph of $\cN_\a$ might be such an interesting example since Proposition \ref{theta N} and Lemma \ref{N aram} imply that $\cvartheta(\cN_\a)=\Aram(K_\a)$.

It remains unknown whether the Lov\'asz number coincides with $C_{\rm{0E}}$ for every classical channel. For any confusability graph $G$,  a variant of  Lov\'asz number called Schrijver number~\cite{Schrijver1979, McEliece1978} was proved to be a tighter upper bound on the entanglement-assisted independence number than the  Lov\'asz number~\cite{Cubitt2014a} . However, it remains unknown whether Schrijver number will converge to the Lov\'asz number in the asymptotic limit. Note that a gap between the Lov\'asz number and the regularized Schrijver number would imply a separation between $C_{\rm{0E}}(G)$ and $\vartheta(G)$. Moreover, it would be interesting to consider how to estimate the regularization of a sequence of semidefinite programs (or linear programs).

\section{Conclusions}
In summary, we have shown that there is a separation between the quantum Lov\'asz  number and the entanglement-assisted zero-error classical capacity.  We have explicitly exhibited a class of quantum channels for which the quantum  Lov\'asz number is strictly larger than the entanglement-assisted zero-error capacity. In particular, we have obtained the reversibility of these channels in the zero-error communication and simulation setting when assisted with quantum no-signalling correlations.

For any classical channel with a positive zero-error capacity, it is known that the feedback-assisted or NS-assisted zero-error capacity are both equal to the fractional packing number. In contrast, for quantum channels, we have shown that the feedback-assisted or the NS-assisted zero-error capacity is not given by the quantum fractional packing number in~\cite{Duan2016}. It also raises a new question to explore other quantum extensions of the fractional packing number.
 
\section*{Acknowledgment}
We would like to thank Andreas Winter for helpful suggestions.   We also thank the referees of AQIS'16 for useful comments which improved the presentation of this paper. This work was partly supported by the Australian Research Council under Grant Nos. DP120103776 and FT120100449.

\end{document}